\documentclass[envcountsame]{llncs}

\usepackage{times}
\usepackage[utf8]{inputenc}
\usepackage[english]{babel}
\usepackage{amsmath}
\usepackage{algorithm}
\usepackage[noend]{algpseudocode}
\usepackage{enumerate}
\usepackage{wasysym}
\usepackage{tikz}
\usepackage{amssymb}
\usepackage{dsfont}

\usetikzlibrary{shapes,positioning}

\newcommand{\nat}{\mathbb{N}}

\newcommand{\val}{\mathsf{val}}
\newcommand{\lex}{<_{\mathsf{lex}}}
\newcommand{\llex}{<_{\mathsf{llex}}}
\newcommand{\lleqx}{\leq_{\mathsf{llex}}}
\newcommand{\lgeqx}{\geq_{\mathsf{llex}}}
\newcommand{\lgex}{>_{\mathsf{llex}}}

\newcommand{\ex}{\mathsf{ex}}
\newcommand{\uo}{\mathsf{uo}}
\newcommand{\ur}{\mathsf{ur}}

\newcommand{\canon}{\mathsf{canon}}

\newcommand{\ranked}{\mathsf{ranked}}
\newcommand{\rooted}{\mathsf{root}}
\newcommand{\rooty}{\mathsf{rooty}}

\newcommand{\Center}{\mathsf{center}}
\newcommand{\dflr}{\mathsf{dflr}}
\newcommand{\PSPACE}{\textsc{pspace}}
\newcommand{\EXPTIME}{\textsc{exptime}}
\newcommand{\ALOGTIME}{\textsc{alogtime}}
\newcommand{\Ptime}{\textsc{ptime}}
\newcommand{\ecc}{\mathsf{ecc}}
\newcommand{\rootToY}{\mathsf{rty}}
\newcommand{\centr}{\mathsf{center}}
\newcommand{\even}{\mathsf{even}}

\begin{document}
\pagestyle{plain}
\title{Compressed Tree Canonization
}
\author{
  Markus Lohrey\inst{1}\and
  Sebastian Maneth\inst{2}\and
  Fabian Peternek\inst{2}
}
\institute{
Universit{\"a}t Siegen, Germany
\and
University of Edinburgh, UK
}
\maketitle
\begin{abstract}
Straight-line (linear) context-free tree (SLT) grammars have been used to compactly
represent ordered trees. It is well known that equivalence of 
SLT grammars is decidable in polynomial time.
Here we extend this result and show that isomorphism of
unordered trees given as SLT grammars is decidable in polynomial time.  
The proof constructs a compressed version of the canonical form
of the tree represented by the input SLT grammar.
The result is generalized to unrooted trees by ``re-rooting'' the compressed trees 
in polynomial time. We further show that bisimulation equivalence
of unrooted unordered trees represented by SLT grammars is decidable in polynomial time. 
For non-linear SLT grammars which can have double-exponential compression ratios,
we prove that unordered isomorphism is $\PSPACE$-hard and in $\EXPTIME$.
The same complexity bounds are shown for bisimulation equivalence.
\end{abstract}

\section{Introduction}

Deciding isomorphism between various mathematical objects is an important topic in theoretical
computer science that has led to intriguing open problems like the precise complexity of the graph 
isomorphism problem. An example of an isomorphism problem, where the knowledge seems to be rather complete,
is tree isomorphism. Aho, Hopcroft and Ullman \cite[page~84]{AhoHoUl74} proved that isomorphism
of unordered trees (rooted or unrooted) can be decided in linear time. An unordered tree is 
a tree, where the children of a node are not ordered. The precise complexity of tree isomorphism
was finally settled by Lindell \cite{Lindell92}, Buss \cite{Bus97}, and Jenner et al. \cite{JKMT03}: 
Tree isomorphism is logspace-complete if the trees are represented by pointer structures 
\cite{Lindell92,JKMT03} and \ALOGTIME-complete if  the trees are represented by expressions \cite{Bus97,JKMT03}.
All these results deal with trees that are given explicity (either by an expression or a pointer structure).
In this paper, we deal with the isomorphism problem for trees that are given in a succinct way. Several succinct encoding schemes
for graphs exist in the literature. Galperin and Wigderson \cite{GaWi83} considered graphs that are given by a boolean
circuit for the adjacency matrix. Subsequent work showed that the complexity of a problem undergoes an
exponential jump when going from the standard input representation to the circuit representation; this phenomenon
is known as upgrading, see \cite{DBLP:conf/lata/DasST14} for more details and references. Concerning graph isomorphism,
it was shown in \cite{DBLP:conf/lata/DasST14}  that its succinct version is \PSPACE-hard, even for very restricted classes
of boolean circuits (DNFs and CNFs).

In this paper, we consider another succinct input representation that has turned out to be more amenable
to efficient algorithms, and, in particular, does not show the upgrading phenomenon known for boolean circuits:
straight-line context-free grammars, i.e., context-free grammars that produce a single object. Such grammars
have been intensively studied for strings and recently also for trees. Using a straight-line grammar,
repeated patterns in an string or tree can be abbreviated by a nonterminal which can be used in different contexts. For strings,
this idea is known as grammar-based compression \cite{CLLLPPSS05,loh12}, and it was extended to trees in \cite{BuLoMa07,LohreyMM13}. In fact 
this approach can be also extended to general graphs by using hyperedge replacement graph grammars; the 
resulting formalism is known as hierarchical graph representation and was studied under an algorithmic perspective
in \cite{LeWa92}. 

The main topic of this paper is the isomorphism problem for trees that are succinctly represented by straight-line
context-free tree grammars. An example of such a grammar contains the productions $S \to A_0(a)$, $A_i(y) \to A_{i+1}(A_{i+1}(y))$
for $0 \leq i \leq n-1$, and $A_n(y) \to f(y,y)$ (here $y$ is called a parameter and in general several parameters may occur in a rule). 
This grammar produces a full binary tree of height $2^n$ and hence has
$2^{2^n+1}-1$ many nodes. This example shows that a straight-line context-free tree grammar may produce a tree, whose
size is doubly exponential in the size of the grammar. The reason for this double exponential blow-up is copying: The parameter
$y$ occurs twice in the right-hand side of the production $A_n(y) \to f(y,y)$. If this is not allowed, i.e., if every parameter 
occurs at most once in every right-hand side, then the grammar is called linear. Straight-line linear (resp., non-linear) context-free tree grammars
are called SLT grammars (resp., ST grammars) in this paper. SLT grammars generalize dags (directed acyclic graphs) that
allow to share repeated subtrees of a tree, whereas SLT grammars can also share repeated patterns that are not complete 
subtrees. 

It turned out that many algorithmic problems are much harder
for trees represented by ST grammars than trees represented by SLT grammars. A good example is the membership problem
for tree automata (\Ptime-complete for SLT grammars \cite{DBLP:journals/jcss/LohreyMS12} and \PSPACE-complete for ST grammars \cite{LoMa06}).
A similar situation arises for the isomorphism problem: We prove that
\begin{itemize}
\item the isomorphism problem for (rooted or unrooted) unordered trees that are given by SLT grammars is \Ptime-complete, and
\item the isomorphism problem for (rooted or unrooted) unordered trees that are given by ST grammars is \PSPACE-hard and in \EXPTIME.
\end{itemize}
Our polynomial time algorithm for SLT grammars constructs from a given SLT grammar $G$ a new SLT grammar $G'$ that produces
a canonical representation of the tree produced by $G$. Our canonical representation of a given rooted unordered tree $t$ is the 
ordered rooted tree (in an ordered tree the children of a node are ordered) 
that has the lexicographically smallest preorder traversal among all ordered versions of $t$.
For unrooted SLT-compressed trees, we first compute a compressed representation of the center node of a given
SLT-compressed unrooted tree $t$. Then we compute an SLT grammar that produces the rooted 
version of $t$ that is rooted in the center node. This is also the standard reduction of the unrooted isomorphism problem 
to the rooted isomorphism problem in the  uncompressed setting, but it requires some work to carry out this reduction
in polynomial time in the SLT-compressed setting. 

Our techniques can be also used to show that checking bisimulation equivalence of trees that are represented by SLT grammars
is \Ptime-complete. This generalizes the well-known \Ptime-completeness of bisimulation for dags \cite{BaGaSa92}. 
In this context, it is interesting to note that bisimulation equivalence for graphs that are given by
hierarchical graph representations is \PSPACE-hard and in \EXPTIME{} \cite{DBLP:conf/concur/BrenguierGS12}.

\section{Preliminaries}

For $k \geq 0$ let $[k] = \{1, \ldots, k\}$.
Let $\Sigma$ be an alphabet.
By $T_\Sigma$ we denote the set of all (ordered, rooted) 
trees over the alphabet $\Sigma$. It is defined recursively as
the smallest set of strings $T$ such that if 
$t_1,\dots,t_k\in T$ and $k\geq 0$ then also $\sigma(t_1,\dots,t_k)$ is in $T$.
For the tree $a()$ we simply write $a$.
The set $D(t)$ of {\em Dewey addresse}s of a tree $t=\sigma(t_1,\dots,t_k)$ is the subset
of $\nat^*$ defined recursively as $\{\varepsilon\}\cup \bigcup_{i\in[k]} i \cdot D(t_i)$.
Thus $\varepsilon$ denotes the root node of $t$ and $u\cdot i$ denotes
the $i$-th child of $u$. 
For $u \in D(t)$, we denote by $t[u]\in\Sigma$ the symbol at $u$, i.e., if 
$t=\sigma(t_1,\dots,t_k)$, then $t[\varepsilon]=\sigma$ and $t[i\cdot u] = t_i[u]$.
The {\em size} of the tree $t$ is $|t| = |D(t)|$.

A {\em ranked alphabet} $N$ is a finite set of symbols each of
which equipped with a non-negative integer, called its ``rank''.
We write $A^{(k)}$ to denote that the rank of $A$ is $k$,
and write $N^{(k)}$ for the set of symbols in $N$ that
have rank $k$. 
For an alphabet $\Sigma$ and a ranked alphabet $N$, we denote
by $T_{N\cup\Sigma}$ the set of trees $t$ over $N\cup \Sigma$ with the
property that if $t[u]=A\in N^{(k)}$, then $u\cdot i\in D(t)$ if and only if $i \in [k]$. 
Thus, if a node is labeled by a ranked
symbol, then the rank determines the number of children of the node. 

We fix a special alphabet $Y=\{y_1,y_2,\dots\}$ of \emph{parameters}. For $y_1$ we also write $y$.
The parameters are considered as symbols of rank zero, and by
$T_{\Sigma\cup N}(Y)$ we denote the set of trees from
$T_{\Sigma \cup N \cup Y}$ where each symbol in $Y$ has rank zero.
We write $Y_k$ for the set of parameters $\{y_1,\dots,y_k\}$. 
For trees $t, t_1, \ldots, t_k \in T_{\Sigma\cup N}(Y)$ we denote
by $t[y_j\leftarrow t_j\mid j\in[k]]$ the tree obtained from $t$ 
by replacing in parallel every occurrence of $y_j$ ($j \in [k]$) by $t_j$.

A {\em context-free tree grammar} is a tuple $G=(N,\Sigma,S,P)$ where
$N$ is a ranked alphabet of nonterminal symbols,
$\Sigma$ is an alphabet of terminal symbols with $\Sigma\cap N=\emptyset$,
$S\in N^{(0)}$ is the start nonterminal, and $P$ is a finite set of productions of the form
$A(y_1,\dots,y_k)\to t$ where $A\in N^{(k)}$, $k\geq 0$, and $t\in T_{N\cup\Sigma}(Y_k)$. 
Occasionally, we consider context-free tree grammars without a start nonterminal.
Two trees $\xi,\xi'\in T_{N\cup\Sigma}(Y)$ are in the one-step derivation relation
$\Rightarrow_G$ induced by $G$, if $\xi$ has a subtree
$A(t_1,\dots,t_k)$ with $A\in N^{(k)}$, $k\geq 0$ such that
$\xi'$ is obtained from $\xi$ by replacing this subtree by
$t[y_j\leftarrow t_j\mid j\in[k]]$, where $A(y_1,\dots,y_k)\to t$ is 
a production in $P$. The tree language $L(G)$ produced by $G$ is
$\{t\in T_\Sigma\mid S\Rightarrow_G^* t\}$. 
We assume that $G$ contains no useless productions, i.e., each production
as applied in the derivation of some terminal tree in $T_\Sigma$. 
The {\em size} of the grammar $G$ is $|G| = \sum_{(A(y_1,\ldots,y_k) \to t) \in P} |t|$.
The grammar $G = (N,\Sigma,S,P)$ is {\em deterministic} if for every
$A \in N$ there is exactly one production of the form $A \to t$.
The  grammar $G$ is {\em acyclic}, if there is a linear order $<$ on $N$ such that
$A < B$ whenever $B$ occurs in a tree $t$ with $(A \to t) \in P$. A 
deterministic and acyclic grammar is called \emph{straight-line}. Note 
that $|L(G)|=1$ for a straight-line grammar. We denote the unique tree $t$ produced
by the straight-line tree grammar $G$ by $\val(G)$. Moreover, for a tree
$t \in T_{\Sigma \cup N}(Y)$  we denote with $\val_G(t) \in T_{\Sigma}(Y)$ the unique tree obtained
from $t$ by applying productions from $G$ until only terminal symbols from $\Sigma$
occur in the tree. If $G$ is clear from the context, we simply write $\val(t)$ for $\val_G(t)$.
The grammar $G$ is \emph{linear} if for every production $(A \to t) \in P$ and every $y \in Y$,
$y$ occurs at most once in $t$.

For a \emph{straight-line linear context-free tree grammar}  we say \emph{SLT grammar}.
For a (not necessarily linear) \emph{straight-line context-free tree grammar} we say \emph{ST grammar}. 
Most of this paper is about SLT grammars, only in Section~\ref{sect:nonlinear}
we deal with (non-linear) ST grammars. SLT grammars generalize rooted node-labelled dags (directed acyclic graph),
where the tree defined by such a dag is obtained by unfolding the dag starting from the root (formally,
the nodes of the tree are the directed paths in the dag that start in the root).
A dag can be viewed as an SLT grammar, where all nonterminals have rank $0$ (the nodes of the dag correspond
to the nonterminal of the SLT grammar). Dags are less succinct than SLT grammars (take the tree $f^N(a)$ for $N = 2^n$),
which in turn are less succinct than general ST grammars (take a full binary tree of height $2^n$).
We need the following fact:

\begin{lemma} \label{lemma:ST-SLT}
A given ST grammar $G$ can be transformed in exponential time into an equivalent SLT grammar.
\end{lemma}

\begin{proof}
In fact, an ST grammar $G$ can be transformed in exponential time into an equivalent dag.
This dag is obtained by viewing the right hand side $t(x_1, \ldots, x_n)$ of a $G$-production 
$A(x_1, \ldots, x_n) \to t(x_1, \ldots, x_n)$ as a dag, by merging for all $i \in [k]$
all $x_i$-labelled leafs into a single $x_i$-labelled node. In this way, $G$ becomes
a so called hyperedge replacement graph grammar (or hierarchical graph definition in the sense
of \cite{LeWa92}) that produces a dag of exponential size,
which can be  constructed in exponential time from $G$, and whose unfolding is $\val(G)$.
\qed
\end{proof}
A \emph{context} is a tree in $T_{\Sigma \cup N}(\{y\})$ with exactly one
occurrence of $y$. We denote with $\mathcal{C}_{\Sigma \cup N}$ the set
of all contexts and write $\mathcal{C}_{\Sigma}$ for the set of contexts 
that contain only symbols from $\Sigma$.
For a context $t(y)$ and a tree $t'$ we write $t[t']$ for $t[y\leftarrow t']$.
Occasionally, we also consider SLT grammars, where the start nonterminal 
belongs to $N^{(1)}$, i.e., has rank $1$. We call such a grammar a \emph{$1$-SLT grammar}.
Note that $\val(G)$ is a context if $G$ is a  $1$-SLT grammar~$G$.

In the literature, SLT grammars are usually  defined over a ranked terminal
alphabets. The following lemma is proved in~\cite{DBLP:journals/jcss/LohreyMS12}; 
the proof immediately carries over to our setting where $\Sigma$ is not ranked. 

\begin{lemma}\label{obs}
One can transform in polynomial
time an SLT grammar into an equivalent SLT grammar, where each production has one of the following
four types (where $\sigma\in\Sigma$ and $A,B,C,A_1,\dots,A_k\in N$):
\begin{enumerate}
\item[(1)] $A\to\sigma(A_1,\dots,A_k)$,
\item[(2)] $A\to B(C)$,
\item[(3)] $A(y)\to\sigma(A_1,\dots,A_i,y,A_{i+1},\dots,A_k)$, or
\item[(4)] $A(y)\to B(C(y))$. 
\end{enumerate}
\end{lemma}
In particular, note that $N$ contains only nonterminals of rank at most $1$.

In the following, we will only deal with SLT grammars $G$ having the property from Lemma~\ref{obs}.
For $i\in[4]$, we denote with $G(i)$ the SLT grammar (without start nonterminal) consisting
of all productions of $G$ of type  $(i)$ from Lemma~\ref{obs}.

\medskip
\noindent
{\bf Region Restrictions.}\quad
A \emph{straight-line program} (SLP) can be seen as a 1-SLT grammar $G=(N,\Sigma,S,P)$
containing only productions of the form $A(y) \to B(C(y))$  and $A(y) \to \sigma(y)$ with 
$B,C \in N$ and $\sigma\in\Sigma$.
Thus, $G$ contains 
ordinary rules of a context-free string grammar in Chomsky normal form
(but written as monadic trees). Intuitively, if 
$\val(G)=a_1(\cdots a_n(y)\cdots)$ then $G$ produces the string $a_1\cdots a_n$
and we also write $\val(G) = a_1\cdots a_n$.
For a string $w=a_1\cdots a_n$ and two numbers $l,r\in[n]$ with $l\leq r$
we denote by $w[l,r]$ the substring $a_l a_{l+1}\cdots a_r$.
The following result is a special case of \cite{Hag00}, where it is shown that a so called
composition system (an SLP extended with right-hand sides of the form $A[l,r]$ for positions $l \leq r$)
can be transformed into an ordinary SLP.

\begin{lemma}\label{lm:region}
For a given SLP $G$ and two binary encoded numbers
$l, r \in[|\val(G)|]$ with $l\leq r$ one can compute in polynomial time an SLP $G'$ such that
$\val(G')=\val(G)[l,r]$.
\end{lemma}

\section{Isomorphism of Unrooted SLT-Represented Trees}

For a tree $t$ we denote with $\mathsf{uo}(t)$  the unordered rooted version
of $t$. It is the node-labeled directed graph $(V, E, \lambda)$ where
$V=D(t)$ is the set of nodes, $$E=\{(u, u\cdot i) \mid i \in \nat, u \in \nat^*, u\cdot i \in D(t)\}$$
is the edge relation, and $\lambda$ is the node-labelling function with $\lambda(u) = t[u]$.
For an SLT grammar $G$, we also write $\val_{\mathsf{uo}}(G)$ for $\mathsf{uo}(\val(G))$.

In this section, we present a polynomial time algorithm for deciding 
$\mathsf{uo}(\val(G_1)) = \mathsf{uo}(\val(G_2))$ for two given
SLT grammars $G_1$ and $G_2$. For this, we will first define 
a canonical representation of a given tree $t$, briefly $\canon(t)$,
such that $\uo(s)$ and $\uo(t)$ are isomorphic if and only if
$\canon(s) = \canon(t)$. Then, we show how to produce for a given SLT grammar
$G$ in polynomial time an SLT grammar for $\canon(\val(G))$.

For reasons that will become clear in a moment we have to restrict
to trees $t \in T_\Sigma$ that have the following property: For all
$u,v \in D(t)$, if $t[u] = t[v]$ then $u$ and $v$ have the same number
of children (nodes with the same label have the same number of children).
Such trees are called {\em ranked trees}. 
For the purpose of deciding the isomorphism problem for unorderd
SLT-represented trees this is not a real restriction. Denote for a tree
$t \in T_\Sigma$ the ranked tree $\ranked(t)$ such that $D(t) = D(\ranked(t))$
and for every $u \in D(t)$ with $t[u]=\sigma$: if $u$ has $k$ children in $t$, then 
$\ranked(t)[u] = \sigma_k$, where $\sigma_k$ is a new symbol. 
Then we have:
\begin{itemize}
\item $\uo(s)$ and $\uo(t)$ are isomorphic if and only if $\uo(\ranked(s))$ and $\uo(\ranked(t))$
are isomorphic.
\item For an SLT grammar $G$ we construct in polynomial time the SLT grammar
$\ranked(G)$ obtained from $G$ by changing every production $A\to t$ into 
$A\to\ranked(t)$, where $\ranked$ is extended to trees 
over $\Sigma$ and nonterminals by defining
$\ranked(t)[u]=t[u]$ if $t[u]$ is a nonterminal. Then we have
$\val(\ranked(G)) = \ranked(\val(G))$.
\end{itemize}
Hence, in the following we will only consider ranked trees, and all SLT grammars will produce 
ranked trees as well.

\subsection{Length-Lexicographical Order and Canons}\label{sect:lex}

Let us fix the alphabet $\Sigma$.
For a tree $t\in T_\Sigma$ we denote by $\dflr(t)$ its depth-first left-to-right
traversal string in $\Sigma^*$. It is defined as
$$\dflr(\sigma(t_1,\dots,t_k))=\sigma\, \dflr(t_1)\cdots \dflr(t_k)$$ for
every $\sigma\in\Sigma$, $k\geq 0$, and $t_1,\dots,t_k\in T_\Sigma$. 
Note that for ranked trees $s$ and $t$ it holds that: 
$\dflr(s) = \dflr(t)$
if and only if 
$s=t$. This is the reason for restricting to ranked trees:
for unranked trees this equivalence fails. For instance, $a(a(a))$ and 
$a(a,a)$ have the same depth-first left-to-right
traversal string $aaa$.

Let $<_\Sigma$ be an order on $\Sigma$; it induces the \emph{lexicographical ordering} $\lex$
on equal-length strings $u,w\in\Sigma^*$ as: $u\lex w$ if and only if
there exist $p,u',w'\in\Sigma^*$ and letters $a,b\in\Sigma$ with $a<_\Sigma b$ such that
$u=pau'$ and $w=pbw'$.
The \emph{length-lexicographical ordering} $\llex$ on $\Sigma^*$ is defined by
$u\llex w$ if and only if (i) $|u|<|w|$ or (ii) $|u|=|w|$ and $u\lex w$.
We extend the definition of $\llex$ to trees $s,t$ over $\Sigma$ by 
$s\llex t$ if and only if $\dflr(s)\llex \dflr(t)$.

\begin{lemma}\label{lm:order}
Let $G,H$ be SLT grammars.
It is decidable in polynomial time whether or not (1) $\val(G)\llex\val(H)$
and (2) whether or not $\val(G)=\val(H)$. 
\end{lemma}

\begin{proof}
Point (2) was shown in \cite{BuLoMa07} by computing from $G, H$ in polynomial time 
SLPs $G', H'$ with $\val(G') = \dflr(\val(G))$ and $\val(H')=\dflr(\val(H))$.
Equivalence of SLPs can be decided in polynomial time; this was proved
independently in~\cite{DBLP:journals/tcs/HirshfeldJM96,DBLP:journals/algorithmica/MehlhornSU97,DBLP:conf/esa/Plandowski94}, 
cf.~\cite{loh12}. 

To show (1), we compute in two single bottom-up runs the numbers $n_1=|\val(G')|$ and $n_2 = |\val(H')|$.
If $n_1 \neq n_2$ we are done; so assume that $n = n_1 = n_2$. Next, we compute the first position for which
the strings $\val(G')$ and $\val(H')$ differ. This is done via binary search
and polynomially many equivalence tests: We compute $m=\lceil n/2\rceil$ and,
using Lemma~\ref{lm:region}, construct SLPs $G_1$ and $G_2$ for $\val(G')[1,m]$ and $\val(G')[m+1,n]$, respectively, and
SLPs $H_1$ and $H_2$ for $\val(H')[1,m]$ and $\val(H')[m+1,n]$, respectively. We proceed with $G_1$ and $H_1$ if
$\val(G_1) \neq \val(H_1)$,  otherwise we proceed with $G_2$ and $H_2$. After $c\leq\lceil\log(n)\rceil$ many
steps we obtain SLPs $G_c,H_c$ representing the first position for which
$\val(G')$ and $\val(H')$ differ. We compute the terminal symbols 
$g,h$ with $\val(G_c)=\{g\}$ and $\val(H_c)=\{h\}$ and
determine whether or not $g<_\Sigma h$.
\qed
\end{proof}
For a tree $t\in T_\Sigma$ we define its \emph{canon} $\canon(t)$ as
the smallest tree $s$ with respect to~$\llex$ such that $\uo(s)$ is isomorphic to $\uo(t)$.
Clearly, if $\canon(t)=t$ then also $\canon(t')=t'$ for every subtree $t'$ of $t$.
Hence, in order to determine $\canon(t)$ for  $t=\sigma(t_1,\dots,t_k)$ 
($\sigma\in\Sigma$, $k\geq 0$) let $c_i = \canon(t_i)$ for $i\in[k]$ and 
let $c_{i_1}\lleqx c_{i_2} \lleqx \dots \lleqx c_{i_k}$ be the length-lexicographically ordered
list of the canons $c_1,\dots,c_k$. 
Then $\canon(t) = \sigma(c_{i_1},\dots,c_{i_n})$. The following lemma can be easily
shown by an induction on the tree structure:

\begin{lemma} \label{lm:iso-canon}
Let $s,t \in T_\Sigma$. Then $\mathsf{uo}(s)$ and $\mathsf{uo}(t)$ are isomorphic 
if and only if $\canon(s) = \canon(t)$.
\end{lemma}

\subsection{Canonizing SLT-Represented Trees}

In the following, we denote a tree $A_1(A_2(\cdots A_n(t) \cdots))$, where
$A_1, A_2, \ldots, A_n$ are unary nonterminals with $A_1 A_2 \cdots A_n(t)$.

\begin{theorem}\label{thm:canon}
From a given SLT grammar $G$ one can construct 
in polynomial time an SLT grammar~$G'$  such that $\val(G')=\canon(\val(G))$.
\end{theorem}

\begin{proof}
Let $G=(N,\Sigma,S, P)$.
We assume that $G$ contains no distinct nonterminals $A_1,A_2 \in N^{(0)}$
such that $\val_G(A_1)=\val_G(A_2)$. This is justified because we can test
$\val_G(A_1)=\val_G(A_2)$ in polynomial time by Lemma~\ref{lm:order} (and 
replace $A_2$ by $A_1$ in $G$ in such a case).
We will add polynomially many new nonterminals to $G$ and change the productions
for nonterminals from $N^{(0)}$ such that 
for the resulting SLT grammar $G'$ we have
$\val_{G'}(Z) = \canon(\val_G(Z))$ for every $Z \in N^{(0)}$.

Consider a nonterminal $Z \in N^{(0)}$ and let $M$ be the 
set of all nonterminals in $G$ that can be reached from $Z$.
By induction, we can assume that $G$ already satisfies
$\val_{G}(A)=\canon(\val_G(A))$ for every $A\in M^{(0)} \setminus \{Z\}$.
We distinguish two cases.

\medskip
\noindent
Case (i). $Z$ is of type~(1) from Lemma~\ref{obs}, i.e., has a production $Z\to\sigma(A_1,\dots,A_k)$. 
Using Lemma~\ref{lm:order} we construct an ordering
$i_1,\dots,i_k$ of $[k]$ such that $\val_G(A_{i_1}) \lleqx \val_G(A_{i_2}) \lleqx \cdots\lleqx\val_G(A_{i_k})$.
We obtain $G'$ by replacing the 
production $Z\to\sigma(A_1,\dots,A_k)$ by
$Z \to\sigma(A_{i_1},\dots,A_{i_k})$ and get
$\val_{G'}(Z) = \canon(\val_G(Z))$.

\medskip
\noindent
Case (ii). $Z$ is of type~(2), i.e., has a production $Z\to B(A)$.
Let $\{S_1,\dots,S_m\}=M^{(0)} \setminus \{Z\}$ be an ordering such that 
\[
\val_G(S_1)\llex \val_G(S_2)\llex \cdots \llex \val_G(S_m).
\]
Note that $A$ is one of these $S_i$.
The sequence $S_1, S_2,\ldots,S_m$ partitions the set of all trees $t$ in $T_\Sigma$ 
into intervals $\mathcal{I}_0, \mathcal{I}_1,\ldots,\mathcal{I}_m$ with 
\begin{itemize}
\item $\mathcal{I}_0=\{t\in T_\Sigma\mid t\llex\val_H(S_1)\}$, 
\item $\mathcal{I}_i=\{t\in T_\Sigma\mid\val_H(S_i)\lleqx t\llex\val_H(S_{i+1})\}$ for $1\leq i< m$, and
\item $\mathcal{I}_m=\{t\in T_\Sigma\mid \val_H(S_m)\lleqx t\}$.
\end{itemize}
Consider the maximal $G(4)$-derivation starting from $B(A)$, i.e.,
$$
B(A) \Rightarrow_{G(4)}^* B_1 B_2 \cdots B_N(A),
$$
where $B_i$ is a typ-(3) nonterminal. Clearly, the number $N$ might be of exponential size, but the 
set $\{ B_1, \ldots, B_N\}$ can be easily constructed.
In order to construct an SLT for $\canon(\val_G(Z))$, it remains to reorder the arguments in right-hand
sides of the type-(3) nonterminals $B_i$. The problem is of course that different occurences of a type-(3)
nonterminal in the sequence $B_1 B_2 \cdots B_N$ have to be reordered in a different way. But we will
show that the sequence $B_1 B_2 \cdots B_N$ can be split into $m+1$ blocks
such that all occurrences of a type-(3) nonterminal in one of these blocks have to be reordered
in the same way.

Let $t_k = \val_G(B_{k} B_{k+1} \cdots B_N(A))$ for $k \in [N]$ and 
$t_{N+1} = \val_G(A)$.
Note that $t_1 = \val_G(Z) \lgex \val_G(S_m)$ and
that $t_{k+1} \llex t_k$ for all $k$.
For $i \in [m]$ let $k_i$ be the maximal position $k \leq N+1$ such that 
$t_k \lgeqx \val_G(S_i)$. Since $t_1 \lgeqx \val_G(S_m) \lgeqx \val_G(S_i)$
this position is well defined. Also note that if $A = S_i$, then
we have $k_i = k_{i-1} = \cdots = k_1 = N+1$. For every $0 \leq i \leq m$,
the interval $[k_{i+1}+1, k_i]$ is the set of all positions $k$ such that
$\val_G(t_k) \in \mathcal{I}_i$. Here we set $k_{m+1} = 0$ and $k_0 = N+1$.
Clearly, the interval $[k_{i+1}+1, k_i]$ might be empty.
The positions $k_0, \ldots, k_m$ can be computed in polynomial time, using binary
search combined with Lemma~\ref{lm:order}. 
To apply the latter, note that for a given position $k$ we can compute in polynomial
time an SLT grammar for the tree $t_k$ using Lemma~\ref{lm:region} for the SLP consisting
of all type-(4) productions that are used to derive $B_1 B_2 \cdots B_N$.

We now factorize the string $B_1 B_2 \cdots B_N$ as 
$B_1 B_2 \cdots B_N = u_m u_{m-1} \cdots u_0$, where 
$u_m = B_1 \cdots B_{k_m-1}$ and
$u_i = B_{k_{i+1}} \cdots B_{k_i-1}$ for $0 \leq i \leq m-1$.
By Lemma~\ref{lm:region} we can compute in polynomial time an SLP $G_i$ 
for the string $u_i$. For the further consideration, we view $G_i$
as a $1$-SLT grammar consisting only of type-(4) productions. 
Note that $\val(G_i)$ is a linear tree, where every node is labelled with a type-(3)
nonterminal. We now add reordered versions of  type-(3) productions to 
$G_i$. Consider a type-(3) production $(C(y) \to 
\sigma(A_1,\dots,A_j,y,A_{j+1},\dots,A_k)) \in P$ where  $C
\in \{ B_1, \ldots, B_N\}$. Then we add to $G_i$ the type-(3) production
\[
C(y) \to \sigma(A_{j_1},\dots,A_{j_\nu},y,A_{j_{\nu+1}},\dots,A_{j_k}),
\]
where $\{j_1,\dots,j_k\}=[k]$ and $0\leq\nu\leq k$ are chosen such that
\begin{enumerate}
\item[(1)] $\val_G(A_{j_1})\lleqx\val_G(A_{j_2})\lleqx\cdots\lleqx\val_G(A_{j_k})$ and
\item[(2)] $\val_G(A_{j_\nu})\lleqx \val_G(S_i) \llex\val_G(A_{j_{\nu+1}})$. 
\end{enumerate}
Note that if $\nu=k$ then condition~(2)
states that $\val_G(A_{j_k})\lleqx \val_G(S_i)$, and
if $\nu=0$ then it states that $\val_G(S_i) \llex\val_G(A_{j_1})$. Also
note that condition (2) ensures that for every tree $t \in \mathcal{I}_i$
we have $\val_G(A_{j_\nu})\lleqx t \llex\val_G(A_{j_{\nu+1}})$. Hence,
$\val_G(\sigma(A_{j_1},\dots,A_{j_\nu},t,A_{j_{\nu+1}},\dots,A_{j_k}))$
is a canon.  The crucial observation now is that the above factorization
$u_m u_{m-1} \cdots u_0$ of $B_1 B_2 \cdots B_N$ was defined in such a way
that for every occurrence of a type-(3) nonterminal $C(y)$ in $u_i$, the parameter
$y$ will be substituted by a tree from $\mathcal{I}_i$ during the derivation from $Z$ 
to $\val_G(Z)$. Hence, we reorder the arguments in the right-hand sides of nonterminal
occurrences in $u_i$ in the correct way to obtain a canon.

We now rename the nonterminals in the SLT grammars $G_i$ (which are now of type (3) and type (4))
so that the nonterminal sets of $G, G_0, \ldots, G_m$ are pairwise disjoint.
Let $X_i(y)$ be the start nonterminal of $G_i$ after the renaming. Then we 
add to the current SLT grammar $G$ the union of all the $G_i$, and replace the production
$Z \to B(A)$ by $Z \to X_m X_{m-1} \cdots X_0(A)$. The construction implies that
$\val_{G'}(Z) = \canon(\val_G(Z))$ for the resulting grammar $G'$.

It remains to argue that the above construction can be carried out in polynomial time.
All steps only need polynomial time in the size of the current SLT grammar. Hence, it 
suffices to show that the size of the SLT grammar is polynomially bounded. 
The algorithm is divided into $|N^{(0)}|$ many phases, where in each phase it enforces $\val_{G'}(Z) = \canon(\val_G(Z))$
for a single nonterminal $Z$. 
Consider a single phase, where $\val_{G'}(Z) = \canon(\val_G(Z))$ is enforced for a nonterminal
$Z$. In this phase, we (i) change the production for $Z$ and (ii) add new type-(3) and type-(4) productions
to $G$ (the union of the $G_i$ above). But the number of these new productions is polynomially bounded in the size
of the initial SLT grammar (the one before the first phase), because the nonterminals introduced
in earlier phases are not relevant for the current phase.
This implies that the additive size increase in each 
phase is bounded polynomially in the size of the initial grammar.
\qed
\end{proof}

\begin{corollary}\label{theo:iso-uo}
The problem of deciding whether 
$\val_{\mathsf{uo}}(G_1)$ and $\val_{\mathsf{uo}}(G_2)$
are isomorphic for given  SLT grammars $G_1$ and $G_2$ is \Ptime-complete.
\end{corollary}

\begin{proof}
Membership in \Ptime{} follows immediately from
Lemma~\ref{lm:order}, Lemma~\ref{lm:iso-canon}, and Theorem~\ref{thm:canon}.
Moreover, \Ptime-hardness already holds for dags, i.e., SLT grammars where all nonterminals
have rank $0$, as shown in \cite{LohreyM13}.
\qed
\end{proof}

\section{Isomorphism of Unrooted Unordered SLT-Represented Trees}

An unrooted unordered tree $t$ over $\Sigma$ can be seen as a 
node-labeled (undirected) graph $t=(V,E,\lambda)$, where
$E\subseteq V\times V$ is symmetric and $\lambda:V\to\Sigma$.
For a node $v$ of $t$ we define the 
eccentricity $\ecc_t(v)=\max_{u \in V}\delta_t(u,v)$ and the diameter
$\diameter(t) = \max_{v\in V}\ecc_t(v)$, 
where $\delta_t(u,v)$ denotes the distance from $u$ to $v$ (i.e.,
the number of edges on the path from $u$ to $v$ in $t$).

Let $t \in T_\Sigma$ be a rooted ordered tree over $\Sigma$ and let 
$t'=\uo(t)=(V,E,\lambda)$ be the rooted unordered tree corresponding to $t$.
The tree $\ur(t')=(V,E\cup E^{-1},\lambda)$ over $\Sigma$ is 
the unrooted version of $t'$.
An unrooted unordered tree $t$ can be represented by an SLT grammar $G$
by forgetting the order and root information present in $G$.
Let $\val_{\ur,\uo}(G) = \ur(\uo(\val(G)))$.

In this section it is proved that isomorphism for unrooted unordered trees $t_1,t_2$
represented by SLT grammars $G_1,G_2$, respectively, can be solved in polynomial time
with respect to $|G_1|+|G_2|$. We reduce the problem to the 
(rooted) unordered case that was solved in Corollary~\ref{theo:iso-uo}.

Let $t=(V,E,\lambda)$ be an unordered unrooted tree.
A node $u$ of $t$ is called \emph{center node of $t$} 
if for all leaves $v$ of $t$:
\[
\delta_t(u,v)\leq(\diameter(s)+1)/2.
\]
Let $\Center(t)$ be the set of all center nodes of $t$.
One can compute the center nodes by deleting all leaves of the tree and 
iterating this step, until the current tree consists of at most two nodes. These
are the center nodes of $t$. In particular, $t$ has either one or two center nodes.
Another characterization of center nodes that is important for our algorithm is via
longest paths. Let $p = (v_0, v_1, \ldots, v_n)$ be a longest simple path in $t$, i.e., 
$n = \diameter(t)$. Then the middle points $v_{\lfloor n \rfloor}$ and $v_{\lceil n \rceil}$
(which are identical if $n$ is even) are the center nodes of $t$. These nodes  are independent
of the concrete longest path $p$.

Note that there are two center nodes if and only if $\diameter(t)$ is odd. Since our constructions
are simpler if a unique center node exists, we first make sure that $\diameter(t)$ is even. 
Let $\#$ be a new symbol not in $\Sigma$. 
For an unrooted unordered tree $t$ we denote by $\even(t)$ the tree where every pair of edge $(u,v), (v,u)$
is replaced by the edges $(u,v'),(v',v), (v,v'), (v',u)$, where $v'$ is a new node labelled $\#$. Then for an
SLT grammar $G=(N,\Sigma,P,S)$
we let $\even(G)=(N,\Sigma\cup\{\#\},P',S)$ be the
SLT grammar where $P'$ is obtained from $P$ by replacing every subtree $\sigma(t_1,\ldots,t_k)$
with $\sigma\in\Sigma$, $k\geq 1$, in a right-hand side by the subtree $\sigma(\#(t_1),\ldots,\#(t_k))$.
Observe that
\begin{itemize}
    \item $\val_{\ur,\uo}(\even(G)) = \even(\val_{\ur,\uo}(G))$,
    \item $\diameter(\even(t)) = 2\cdot\diameter(t)$ is even, i.e., $\even(t)$ has only one center node,
        and
    \item trees $t$ and $s$ are isomorphic if and only if $\even(t)$ and $\even(s)$ are isomorphic.
\end{itemize}
Since $\even(G)$ can be constructed in polynomial time, we assume in the following that every SLT grammar
produces a tree of even diameter and therefore has only one center node.
For a tree $t$ of even diameter, we denote with $\Center(t)$ its unique center node.

Let $u\in V$. We construct a rooted version $\rooted(t,u)$ of $t$,
with root node $u$. 
We set $\rooted(t,u) = (V,E',\lambda)$, where
$E' = \{ (v,v') \in E \mid  \delta_t(u,v) < \delta_t(u,v') \}$.

Two unrooted unordered trees $t_1,t_2$ of even diameter are isomorphic if and only if
$\rooted(t_1,\Center(t_1))$ is isomorphic to $\rooted(t_2,\Center(t_2))$.
Thus, we can solve in polynomial time the isomorphism problem for unrooted unordered trees
represented by SLT grammars $G,G'$ by
\begin{enumerate}
\item[(1)] determining in polynomial time compressed representations
$\tilde{u}_1$ and $\tilde{u}_2$  of $u_1=\Center(\val_{\mathsf{ur,uo}}(G))$ and $u_2=\Center(\val_{\mathsf{ur,uo}}(G'))$, 
respectively  (Section~\ref{sec:find-center}),
\item[(2)] constructing in polynomial time SLT grammars $G_1,G_2$ 
    such that 
    $\val_{\mathsf{uo}}(G_1)=\rooted(\val_{\mathsf{ur,uo}}(G),u_1)$ and
    $\val_{\mathsf{uo}}(G_2)=\rooted(\val_{\mathsf{ur,uo}}(G'),u_2)$  (Section~\ref{sec:reroot}), and 
\item[(3)] testing in polynomial time if
$\val_{\mathsf{uo}}(G_1)$ is isomorphic to
$\val_{\mathsf{uo}}(G_2)$ (Corollary~\ref{theo:iso-uo}).
\end{enumerate}

\subsection{Finding Center Nodes} \label{sec:find-center}

Let $G=(N,\Sigma,S,P)$ be an SLT grammar.
A \emph{$G$-compressed path} $p$ is a string of pairs
$p=(A_1,u_1)\cdots(A_n,u_n)$ such that
for all $i\in[n]$, $A_i\in N$, $A_1 = S$,
 $u_i\in D(t_i)$ is a Dewey address in $t_i$ where $(A_i \to t_i) \in P$,
$t_i[u_i] = A_{i+1}$ for $i < n$, 
and $t_i[u_n]\in\Sigma$. If we omit  the condition $t_i[u_n]\in\Sigma$,
then $p$ is a partial $G$-compressed path.
Note that by definition, $n\leq |N|$.
A partial $G$-compressed path uniquely represents one particular
node in the derivation tree of $G$, and a $G$-compressed path represents
a leaf of the derivation tree and hence a node of $\val(G)$. We denote this node by
$\val_G(p)$. The concatenation $u_1, u_2,\ldots, u_n$ of the Dewey addresses is denoted by $u(p)$.

For a context $t(y) \in \mathcal{C}_\Sigma$  we define
$\ecc(t) = \ecc_t(y)$ (recall that in a context there is a unique occurence of the parameter $y$) and
$\rootToY(t) =\delta_t(\varepsilon, y)$ (the distance from the root to the parameter $y$). For a tree
$s \in T_\Sigma$ we denote with $h(s)$ its height.
We extend these notions to contexts $t \in \mathcal{C}_{\Sigma\cup N}$ and
trees $s \in T_{\Sigma\cup N}$ by
$\ecc(t) =\ecc(\val_G(t))$,  $\rootToY(t) =\rootToY(\val_G(t))$, and
$h(s) = h(\val_G(s))$. 

Eccentricity, distance from root to $y$, and height can be computed in
polynomial time for all nonterminals bottom-up.
To do so, observe that for two contexts $t(y),t'(y) \in  \mathcal{C}_{\Sigma\cup N}$ 
and a tree $s \in T_{\Sigma\cup N}$ 
we have
\begin{itemize}
    \item $\rootToY(t[t']) = \rootToY(t) + \rootToY(t')$,
    \item $\ecc(t[t']) = \max\{\ecc(t'), \ecc(t)+\rootToY(t')\}$, and
    \item $h(t[s]) = \max\{h(s),\rootToY(t) + h(s)\}$.
\end{itemize}
Similarly, for a context $t(y) = \sigma(s_1, \ldots s_i, y, s_{i+1}, \ldots, s_k)$ 
and a tree $s = \sigma(s_1, \ldots, s_k)$ 
we have:
\begin{itemize}
    \item $\rootToY(t) = 1$,
    \item $\ecc(t) = 2+\max\{ h(s_i) \mid 1 \leq i \leq k\}$, and
    \item $h(s) = 1+\max\{ h(s_i) \mid 1 \leq i \leq k\}$.
 \end{itemize}
 Finally, note that for the tree $t[s]$ ($t(y) \in \mathcal{C}_{\Sigma}, s\in T_{\Sigma}$)
 we have
 \begin{equation} \label{eq:diameter}
 \diameter(t[s]) = \max \{ \diameter(t), \diameter(s), \ecc(t) + h(t) \}.
\end{equation}
Our search for the center node of an SLT-compressed tree is based on the following lemma.
For a context $t(y) \in \mathcal{C}_{\Sigma}$, where $u$ is the Dewey address of the parameter $y$,
and a tree $s \in T_{\Sigma}$ we say that a node $v$ of $t[s]$ belongs to $t$ if the Dewey
address of $v$ is in $D(t) \setminus \{ u \}$.  Otherwise, we say that $v$ belongs to 
$s$, which means that $u$ is a prefix of the Dewey address of $v$.

\begin{lemma} \label{lemma:search-center}
Let $t(y) \in \mathcal{C}_{\Sigma}$ be a context and $s\in T_{\Sigma}$ a tree
such that $\diameter(t[s])$ is even. Let $c = \Center(t[s])$.
Then we have the following:
\begin{itemize}
\item If $\ecc(t) \leq h(s)$ then $c$ belongs to $s$.
\item If $\ecc(t) > h(s)$ then $c$ belongs to $t$.
\end{itemize}
\end{lemma}

\begin{proof}
Let us first assume that $\ecc(t) \leq h(s)$, Then we have $\diameter(t) \leq 2 \cdot \ecc(t) \leq \ecc(t)+h(s)$, i.e., 
$\diameter(t[s]) = \max\{\diameter(s),\ecc(t)+h(s)\}$ by \eqref{eq:diameter}. Together with $\ecc(t) \leq h(s)$ this 
implies that the middle point of a longest path in $s[t]$ (which is $c$) belongs to the tree $s$.

Next, assume that $\ecc(t) = h(s)+1$. Then we have $\diameter(s) \leq 2 \cdot h(s) < \ecc(t)+h(s)$, i.e.,
$\diameter(t[s]) = \max\{\diameter(t),\ecc(t)+h(s)\}$.   Moreover, we claim that $\ecc(t)+h(s) \geq \diameter(t)$.
In case $\diameter(t) = \ecc(t)$, this is clear. Otherwise, $\diameter(t) > \ecc(t)$ and a longest path in $t$
does not end in the parameter node $y$. It follows that $\diameter(t) \leq 2 \cdot (\ecc(t) - 1) < \ecc(t) + h(s)$.
Thus, we have $\diameter(t[s]) = \ecc(t)+h(s) = 2 \cdot h(s)+1$, which is odd, a contradiction. Hence, this 
case cannot occur.

Finally, assume that $\ecc(t) > h(s)+1$. Again, we get $\diameter(t[s]) = \max\{\diameter(t),\ecc(t)+h(s)\}$.
Moreover, since $\ecc(t) > h(s)+1$ the  center nodes $c$ must belong to $t$.
\qed
\end{proof}

\begin{lemma}   \label{lem:center_points}
For a given SLT grammar $G$ such that $\val_{\mathsf{ur,uo}}(G)$ has even diameter,
one can construct a $G$-compressed path for $\centr(\val_{\mathsf{ur,uo}}(G))$.    
\end{lemma}

\begin{proof}
Consider the recursive Algorithm~\ref{alg:center_point}.
It is started with $t_l=y$, $t_r=p=\varepsilon$ and $A=S$ and 
computes the node $\centr(\val_{\mathsf{ur,uo}}(G))$. 
The following invariants are preserved by the algorithm:
If $\centr(t_l, A, t_r, p)$ is called, then we have:
\begin{itemize}
\item If $A$ has rank $0$ then $t_r = \varepsilon$
\item $\val(G) = \val(t_l[A[t_r]])$ (here we set $t[\varepsilon]=t$).
\item The tree $t_l[A[t_r]]$ can be derived from the start variable $S$.
\item $p$ is the partial $G$-compressed path to the distinguished $A$ in $t_l[A[t_r]]$. 
\item $\centr(\val_{\mathsf{ur,uo}}(G))$ belongs to the subcontext $\val(A)$ in
$\val(t_l)[\val(A)[ \val(t_r)]]$.
\end{itemize}
For a call $\centr(t_l, A, t_r, p)$, the algorithm distinguishes on the right-hand side of $A$.
If this right-hand side has the form $A(B)$ or $A(B(y))$, then, by comparing
$\ecc(t_l[B(y)])$ and $h(C[t_r])$, we determine,
whether the search for the center node has to continue in $B$ or $C$, see Lemma~\ref{lemma:search-center}.

The case that the right-hand side of $A$ has the form $\sigma(A_1,\ldots, A_k)$ is a bit 
more complicated. Let $s_l = \val(t_l)$ and $s_i = \val(A_i)$ (by the first invariant we know that $t_r = \varepsilon$).
We have to find the center node of $t := s_l(\sigma(s_1, \ldots, s_k)$ and by the last invariant we
know that it is contained in $\sigma(s_1, \ldots, s_k)$. We now consider all $k$ many cuts of $t$ along one of the edges
between the $\sigma$-node and one of the $s_i$, i.e., we cut $t$ into
$s_l(\sigma(s_1, \ldots, s_{i-1}, y, s_{i+1}, \ldots, s_k)$ and $s_i$. Using again Lemma~\ref{lemma:search-center},
it suffices to compare $\ecc(s_l(\sigma(s_1, \ldots, s_{i-1}, y, s_{i+1}, \ldots, s_k)))$ and $h(s_i)$ in order to
determine whether the center node belongs to $s_l(\sigma(s_1, \ldots, s_{i-1}, y, s_{i+1}, \ldots, s_k)$
or $s_i$. If for some $i$, it turns out that the center node is in $s_i$, then we continue the search with $A_i$. 
Finally, assume that for all $i$, it turns out that the center node is in 
$s_l(\sigma(s_1, \ldots, s_{i-1}, y, s_{i+1}, \ldots, s_k)$. Since by the last invariant, the
center node is in $\sigma(s_1, \ldots, s_k)$, the $\sigma$-labelled
 node must be the center node. The case of a production 
 $A(y) \rightarrow \sigma(A_1,\ldots A_{s-1},y, A_{s+1},\ldots,A_k)$  can be dealt with similarly.
 
 Note that $|t_l|+|t_r|$ stays bounded by the size of $G$. Hence,
 whenever $\ecc(t)$ and $h(t)$ have to be determined by the algorithm, then $t$
 is a polynomial size  tree build from terminal and nonterminal symbols. By the previous remarks,
 $\ecc(t)$ and $h(t)$ can be computed in polynomial time.
\qed
\end{proof}

\begin{algorithm}[t]
    \begin{algorithmic}
        \Procedure {$\centr$}{$t_l, A, t_r, p$}
        \If {$A \rightarrow B(C)$ (and thus $t_r = \varepsilon$) or $A(y) \rightarrow B(C(y))$}
             \If {$\ecc(t_l[B(y)]) \leq h(C[t_r])$}
                    \State  \textbf{return} $\centr(t_l[B(y)], C, t_r, p \cdot (A,1))$ 
               \Else
                     \State \textbf{return} $\centr(t_l, B, C[t_r], p \cdot (A,\varepsilon))$ 
               \EndIf
        \EndIf
        \If {$A \rightarrow \sigma(A_1,\ldots,A_k)$ (and thus $t_r = \varepsilon$)}
             \State $t_i \gets t_l[\sigma(A_1, \ldots, A_{i-1},y,A_{i+1},\ldots,A_k)]$ for all $i \in [k]$
             \If {there is an $i \in [k]$ such that $\ecc(t_i) \leq h(A_i)$}
                 \State  \textbf{return} $\centr(t_i, A_i, \varepsilon, p \cdot (A,i))$ 
              \Else
                 \State \textbf{return} $(p \cdot (A,\varepsilon))$
              \EndIf
           \EndIf
           \If{$A(y) \rightarrow \sigma(A_1,\ldots A_{s-1},y, A_{s+1},\ldots,A_k)$}
                \State $t_i \gets t_l[\sigma(A_1, \ldots, A_{i-1},y,A_{i+1},\ldots, A_{s-1}, t_r, A_{s+1}, \ldots, A_k)]$ if $i<s$
                \State $t_i \gets t_l[\sigma(A_1, \ldots, A_{s-1},t_r,A_{s+1},\ldots, A_{i-1}, y, A_{i+1}, \ldots, A_k)]$ if $s<i$
             \If {there is an $i \in [k]  \setminus \{s\}$ such that $\ecc(t_i) \leq h(A_i)$}
                 \State  \textbf{return} $\centr(t_i, A_i, \varepsilon, p \cdot (A_i,i))$ 
              \Else
                 \State \textbf{return} $(p \cdot (A,\varepsilon))$
              \EndIf
           \EndIf
        \EndProcedure
    \end{algorithmic}
    \caption{Recursive procedure to find the $G$-compressed path for the center node}
    \label{alg:center_point}
\end{algorithm}

\subsection{Re-Rooting of SLT Grammars} \label{sec:reroot}

Let $G=(N,\Sigma, S,P)$ be an SLT grammar (as usual, having the normal form 
from Lemma~\ref{obs})
and $p$ a $G$-compressed path.
Let $s(p) \in T_{\Sigma\cup N}$ be the tree defined inductively as follows:
Let $(A \to t) \in P$ and $u \in D(t)$. Then $s((A,u))=t$.
If $p = (A,t) p'$ with $p'$ non-empty, then either (i) $u=\varepsilon$ and $t = B(C)$  or 
(ii) $u = i \in \nat$ and $t[i] \in N^{(0)}$. In case (i) we set
$s(p) = s(p')[C]$, in case (ii) we set $s(p) = t'[s(p')]$, where $t'(y)$ is obtained from $t$
by replacing the $i$-th argument of the root by $y$.
Note that $s(p')\in \mathcal{C}_{\Sigma\cup N}(\{y\})$ if $p'$ starts with a nonterminal of rank $1$.
Let $s = s(p)$; its size is bounded by the size of $G$.
Note that $s[u(p)]$ is a terminal symbol (recall that $u(p)$ denotes the concatenation of the 
Dewey addresses in $p$). Assume that $s[u(p)]=\sigma\in\Sigma$.
Let $\#$ be a fresh symbol and let $s'$  be obtained from $s$ by changing the label at $u(p)$ 
from $\sigma$ to $\#$.
Let $s'\Rightarrow_G^* s''$ be the shortest derivation
such that $s''[\varepsilon]=\delta\in\Sigma$ (it consists of at most $|N|$ derivation steps).
We denote the $\#$-labeled node in $s''$ by $u$. Finally, let $t$ be obtained from $s''$
by changing the unique $\#$ into $\sigma$. 
We define the \emph{$p$-expansion} of $G$, denoted $\ex_G(p)$, 
as the tuple $(t,u,\sigma,\delta)$. 
Note that $\val_G(p)$ is the unique $\#$-labelled node in $\val_G(s'')$.
Moreover, the  $p$-expansion can be computed in polynomial time from
$G$ and $p$.

The $p$-expansion $(t,u,\sigma,\delta)$ 
has all information needed to construct a grammar $G'$ 
representing the rooted version at $p$ of $\val(G)$.
If $u=\varepsilon$ then also $\val_G(p)=\varepsilon$. Since $G$ is
already rooted at $\varepsilon$ nothing has to be done in this 
case and we return $G'=G$.
If $u\not=\varepsilon$ then $\val_G(p)\not=\varepsilon$ and hence
$t$ contains two terminal nodes which uniquely represent the 
root node and the node $\val_G(p)$ of the tree $\val(G)$.

Let $s_1\in T_\Sigma$ be a rooted ordered tree representing
the unrooted unordered tree $\tilde{s}_1=\ur(\uo(s_1))$.
Let $u\not=\varepsilon$ be a node of $s_1$.
Let $s_1[\varepsilon]=\delta\in\Sigma$ and $s_1[u]=\sigma\in\Sigma$.
A rooted ordered tree $s_2$ that represents the rooted unordered
tree $\tilde{s}_2=\rooted(\tilde{s}_1,u)$ can be defined as follows:
Since $u\neq\varepsilon$, we can write 
$$
s_1=\delta(\zeta_1,\ldots, \zeta_{i-1},  t'[\sigma(\xi_1,\dots,\xi_m)], \zeta_{i+1},\ldots, \zeta_k),
$$
where $t'$ is a context, and $u = i u'$, where $u'$ is the Dewey address of the parameter $y$
in $t'$. We can define $s_2$ as
$$
s_2=\sigma(\xi_1,\dots,\xi_m,\rooty(t')[\delta(\zeta_1,\dots,\zeta_{i-1},\zeta_{i+1},\dots,\zeta_k)]),
$$
where $\rooty$ is a function mapping contexts to contexts 
defined recursively as follows, where $f\in\Sigma$,
$t_1,\dots, t_{i-1}, t_{i+1}, \ldots, t_\ell\in T_\Sigma$, and $t(y), t'(y)\in \mathcal{C}_\Sigma$:
\begin{eqnarray}
\rooty(y) & = & y   \label{rooty-1} \\
\rooty(f(t_1, \ldots, t_{i-1},y,t_{i+1},\ldots, t_\ell)) &=& f(t_1,\dots,t_{i-1},y,t_{i+1},\dots,t_\ell) \label{rooty-2} \\
\rooty( t[t'(y)] ) & = & \rooty(t')[\rooty(t(y))] \label{eq:yroot}
\end{eqnarray}
Intuitively, the mapping $\rooty$ unroots a context $t(y)$ towards 
its $y$-node $u$, i.e., it reverses the path from the root to $u$. 
Thus, for instance,
$\rooty(f(a,y,b))=f(a,y,b)$ and
$\rooty(f(a,g(c,y,d),b))=g(c,f(a,y,b),d)$.

\begin{lemma}\label{lm:unroot}
From a given SLT grammar $G$  and a $G$-compressed path $p$ one
can construct in polynomial time an SLT grammar $G'$ such that
$\val_{\uo}(G')$ is isomorphic to $\rooted(\val_{\ur,\uo}(G),\val_G(p))$.
\end{lemma}

\begin{proof}
Let $G=(N,\Sigma,S,P)$ and $\ex_G(p)=(t,u,\sigma,\delta)$. 
If $u=\varepsilon$ then define $G'=G$.
If $u\not=\varepsilon$ then we can write 
\begin{equation} \label{splitting-of-t}
t = \delta(B_1, \ldots, B_{i-1}, t'[ \sigma(\xi_1,\dots,\xi_m)], B_{i+1},\ldots, B_k),
\end{equation}
where $B_j \in N^{(0)}$, $\xi_j \in T_{N}$,
$t'$ is a context composed of nonterminals $A \in N^{(1)}$ and 
contexts $f(\zeta_1, \ldots, \zeta_{j-1}, y, \zeta_{j+1},\ldots, \zeta_l)$
($f \in \Sigma$, $\zeta_j \in T_{N}$),  and $u = i u'$, where $u'$ is the Dewey address of the parameter $y$
in $t'$.

We define $G'=(N\uplus N',\Sigma,S,P')$ where $N'=\{A'\mid A\in N^{(1)}\}$.
To define the production set $P'$, we extend the definition of $\rooty$ to contexts 
from $\mathcal{C}_{\Sigma \cup N}$ by~(\emph{i}) 
allowing in the trees $t_j$ from Equation~\eqref{rooty-2}  also nonterminals, and 
(\emph{ii})~defining for every $B \in N^{(1)}$, $\rooty(B(y)) = B'(y)$.
We now define the set of productions $P'$ of $P$ as follows:
We put all productions from $P$ except for the start production $(S \to s) \in P$ into
$P'$. For the start variable $S$ we add to $P'$ the production 
$$
S\to \sigma(\xi_1,\dots,\xi_m,\rooty(t')[\delta(B_1,\dots,B_{i-1},B_{i+1},\ldots,B_k)]).
$$
Moreover, let $A\in N^{(1)}$ and $(A(y) \to \zeta) \in P$.
If this is a type-(3) production, then we add 
$A'(y)\to\zeta$ to $P'$.
If $\zeta=B(C(y))$ then add $A'(y)\to C'(B'(y))$ to $P'$.

\medskip
\noindent
\emph{Claim:}
Let $A\in N^{(1)}$. Then $\val_{G'}(A')=\rooty(\val_G(A))$.

\medskip
\noindent
The claim is easily shown by induction on the reverse hierarchical structure of $G$:
Let $(A \to t_A) \in P$. If $t_A =f(A_1,\dots,A_j,y,A_{j+1},\dots,A_l)$ then $\rooty(\val_G(A))=\val_G(A)$. 
Since $(A' \to t_A) \in P'$ and $G'$ contains all productions of $G$ except for the start production,
we obtain $\val_{G'}(A')=\rooty(\val_G(A))$.
If $t_A=B(C(y))$ then, by Equation~\eqref{eq:yroot},
$\rooty(\val_G(B(C(y))))=\rooty(\val_G(C))[\val_G(B)]$.
By induction the latter is equal to $\val_{G'}(C')[\val_{G'}(B')]$
which equals $\val(A')$ by the definition of the right-hand side of $A'$. 
This proves the claim.

\medskip

\noindent
The above claim implies that $\val_{G'}(\rooty(c(y)))=\rooty(\val_G(c(y)))$ for every context $c(y)$ that is composed of 
 contexts $f(\zeta_1, \ldots, \zeta_{j-1}, y, \zeta_{j+1},\ldots, \zeta_l)$
($\zeta_j \in T_{N}$) and nonterminals $A \in N^{(1)}$.
In particular, $\val_{G'}(\rooty(t')) = \rooty(\val_G(t'(y)))$ for the context $t'$ from Equation~\eqref{splitting-of-t}.
Hence, with $s_j=\val_{G'}(\xi_j)=\val_G(\xi_j)$
and $t_j = \val_{G'}(B_j) = \val_G(B_j)$ we obtain
\begin{eqnarray*}
\val(G') &=& \val_{G'}(\sigma(\xi_1,\dots,\xi_m,\rooty(t')[\delta(B_1,\dots,B_{i-1},B_{i+1},\ldots,B_k)])) \\
& = & \sigma( s_1,\dots, s_m, \val_{G'}(\rooty(t'))[ \delta(t_1,\dots,t_{i-1},t_{i+1},\ldots,t_k) ]) \\
& = & \sigma( s_1,\dots, s_m, \rooty(\val_{G}(t'))[ \delta(t_1,\dots,t_{i-1},t_{i+1},\ldots,t_k) ]) .
\end{eqnarray*}
Since $\val(G) = \delta(t_1, \ldots, t_{i-1}, \val_G(t')[ \sigma(s_1,\dots,s_m)], t_{i+1},\ldots, t_k)$, it follows
that $\val_{\mathsf{uo}}(G')$ is isomorphic to $\rooted(\val_{\mathsf{ur,uo}}(G),\val_G(p))$.
\qed
\end{proof}

\begin{corollary} \label{coro-ur-uo-iso}
The problem of deciding whether 
$\val_{\ur,\uo}(G_1)$ and $\val_{\ur,\uo}(G_2)$
are isomorphic for given  SLT grammars $G_1$ and $G_2$ is \Ptime-complete.
\end{corollary}

\begin{proof}
The upper bound follows from  Lemma~\ref{lem:center_points}, Lemma~\ref{lm:unroot}, and Corollary~\ref{theo:iso-uo}.
Hardness for \Ptime{} follows from the \Ptime-hardness for dags \cite{LohreyM13} and the fact that isomorphism of rooted unordered
trees can be reduced to isomorphism of unrooted unordered trees by labelling the roots with a fresh symbol.
\qed
\end{proof}

\section{Bisimulation on SLT-compressed trees}

\newcommand{\bcanon}{\mathsf{bcanon}}

Fix a set $\Sigma$ of node labels.
Let $G = (V,E,\lambda)$ be a directed node-labelled graph, i.e., $E \subseteq V \times V$ is the edge
relation and $\lambda : V \to \Sigma$ is the labelling function. A binary relation $R \subseteq V \times V$
is a bisimulation on $G$, if for all $(u,v) \in R$ the following three conditions hold:
\begin{itemize}
\item $\lambda(u) = \lambda(v)$
\item If $(u,u') \in E$, then there exists $v' \in V$ such that $(v,v') \in E$ and $(u',v') \in R$.
\item If $(v,v') \in E$, then there exists $u' \in V$ such that $(u,u') \in E$ and $(u',v') \in R$.
\end{itemize}
Let the relation $\sim$ be the union of all bisimulations on $G$. It is itself a bisimulation (and hence the largest
bisimulation) and an equivalence
relation. Two rooted unordered trees $s, t$ with node labels from $\Sigma$ and roots $r_s, r_t$
are bisimulation equivalent
if $r_s \sim r_t$ holds in the disjoint union of $s$ and $t$. For instance, the trees $f(a,a,a)$ and $f(a,a)$ are bisimulation
equivalent but the trees 
$f(g(a),g(b))$ and $f( g(a,b))$ are not.

For a rooted unordered tree $t$ we define the bisimulation canon $\bcanon(t)$ inductively as follows:
 Let $t = f(t_1,\ldots, t_n)$ ($n \geq 0$) and let $b_i = \bcanon(t_i)$. Let $s_1, \ldots, s_m$ be a list of trees
such that (i) for every $i \in [m]$, $s_i$ is isomorphic to one of the $b_j$, and  (ii) for every $i \in [n]$ there
is a unique $j \in [m]$ such that $s_i$ and $b_j$ are isomorphic as rooted unordered trees.
Then $\bcanon(t) = f(s_1,\ldots, s_m)$.
In other words: Bottom-up, we eliminate repeated subtrees among the children of a node.
For instance, $\bcanon(f(a,a,a)) = f(a) = \bcanon(f(a,a))$.
The following lemma can be shown by a straightforward induction on the height of trees.

\begin{lemma} \label{lemma:bcanon}
Let $s$ and $t$ be rooted unordered trees. Then $s$ and $t$ are bisimulation equivalent if
and only if $\bcanon(s)$ and  $\bcanon(t)$ are isomorphic.
\end{lemma}
The proof of the following theorem is similar to those of Theorem~\ref{thm:canon}.

\begin{theorem}\label{thm:bcanon}
From a given SLT grammar $G$ one can compute a new SLT grammar $G'$ such that
$\val_{\uo}(G')$ is isomorphic to $\bcanon(\val_{\uo}(G))$.
\end{theorem}

\begin{proof}
Let $G=(N,\Sigma,S, P)$.
We will add polynomially many new nonterminals to $G$ and change the productions
for nonterminals from $N^{(0)}$ such that
for the resulting SLT grammar $G'$ we have
$\uo(\val_{G'}(Z)) = \bcanon(\uo(\val_G(Z)))$ for every $Z \in N^{(0)}$.

Consider a nonterminal $Z \in N^{(0)}$ and let $M$ be the
set of all nonterminals in $G$ that can be reached from $Z$.
By induction, we can assume that $G$ already satisfies
$\uo(\val_{G}(A))=\bcanon(\uo(\val_G(A)))$ for every $A\in M^{(0)} \setminus \{Z\}$.
Moreover, we can assume that $G$ contains no distinct nonterminals $A_1,A_2 \in N^{(0)}$
such that $\uo(\val_G(A_1))$ and $\uo(\val_G(A_2))$ are isomorphic. This is justified because
by Corollary~\ref{theo:iso-uo} we can test in polynomial time whether
$\uo(\val_G(A_1))$ and $\uo(\val_G(A_2))$ are isomorphic
and  replace $A_2$ by $A_1$ in $G$ in such a case (the tree produced by the new grammar is
isomorphic to $\uo(\val(G))$).
Similarly, if there is a type-(1) production $A\to\sigma(A_1,\dots,A_k)$
such that $A_i = A_j$ for $i < j$, then we remove $A_j$ from the parameter list,
and the same is done for type-(3) productions. These preprocessing steps do
not change the bisimulation canon.
We now distinguish two cases.

\medskip
\noindent
Case (i). $Z$ is of type~(1), i.e., has a production $Z\to\sigma(A_1,\dots,A_k)$.
By the above preprocessing, we already have $\uo(\val_G(Z)) = \bcanon(\uo(\val_G(Z)))$,
so nothing has to be done.

\medskip
\noindent
Case (ii). $Z$ is of type~(2), i.e., has a production $Z\to B(A)$.
For $C \in M^{(0)}$ let $n_C = |\val_G(C)|$ and let

$$J = \{  n_C \mid C \in M^{(0)} \setminus \{Z\}\}.$$
We can compute this
set of numbers easily in a bottom-up fashion.

Consider the maximal $G(4)$-derivation starting from $B(A)$, i.e.,
$$
B(A) \Rightarrow_{G(4)}^* B_1 B_2 \cdots B_N(A),
$$
where $B_i$ is a typ-(3) nonterminal.
Let $t_k = \val_G(B_{k} B_{k+1} \cdots B_N(A))$ for $k \in [N]$ and
$t_{N+1} = \val_G(A)$.
For a given position we
can compute in polyomial time the size $|t_i|$ by first computing
an SLT grammar for $t_i$ and then computing the size of the generated
tree bottom-up. Clearly, the sequence $|t_1|, |t_2|, \ldots, |t_{N+1}|$
is monotonically decreasing. This allows to compute, using binary search,
the set of positions
$$
I = \{ i  \mid i \in [N+1],  |t_i| \in J \}.
$$
Note that $|I| \leq |M^{(0)} \setminus \{Z\}|$ and $N+1 \in I$.
Next, we check in polynomial time, using  Corollary~\ref{theo:iso-uo},
for every position $i \in I$, whether $\uo(t_i)$ is isomorphic to $\uo(\val_G(C))$
for some $C \in M^{(0)} \setminus \{Z\}$.
If such a $j$ exists then we keep $i$ in the set $I$, otherwise
we remove $i$ from $I$. After this step, $I$ contains exactly those positions
$i \in I$ such that  $\uo(t_i)$ is isomorphic to $\uo(\val_G(C))$
for some $C \in M^{(0)} \setminus \{Z\}$.

Assume that  $I = \{ i_1, \ldots, i_k\}$ with $1 \leq i_1 < i_2 < \cdots < i_{k-1} < i_k =N+1$.
We now factorize the string $B_1 B_2 \cdots B_N$ as
$$B_1 B_2 \cdots B_N = u_1 B_{i_1-1} u_2 B_{i_2-1} \cdots  u_k B_{i_{k}-1},$$
where $u_j = B_{i_{j-1}} \cdots B_{i_j-2}$ for $j \in [k]$ (set $i_0 =1$).
By Lemma~\ref{lm:region} we can compute in polynomial time an SLP $G_j$
for the string $u_j$. Moreover, we can compute the nonterminals $B_{i_j-1}$ in
polynomial time. For the further consideration, we view $G_j$
as a $1$-SLT grammar consisting only of type-(4) productions.
Note that $\val(G_j)$ is a linear tree, where every node is labelled with a type-(3)
nonterminal.

We now rename the nonterminals in the SLT grammars $G_j$
so that the nonterminal sets of $G, G_1, \ldots, G_k$ are pairwise disjoint.
Let $X_j(y)$ be the start nonterminal of $G_j$ after the renaming. Then we
add to the current SLT grammar $G$ the union of all the $G_j$. Moreover, for
every $j \in [k]$ we add a new nonterminal $C_j$ to $G$, whose right-hand side
is derived from the right-hand side of $B_{i_j-1}$ as follows:
Let the right-hand side for $B_{i_j-1}$ be
$\sigma(A_1,\dots,A_l,y)$ (we can assume that the parameter occurs at the last argument
position, since this is not relevant for the bisimulation canon). We now check whether there
exists an $A_i$ ($i \in [l]$) such that $\uo(\val_G(A_i))$ is isomorphic to $\uo(t_{i_j})$.
If such an $i$ exists then by our preprocessing it is unique, and we
add to $G$ the production $C_j(y) \to \sigma(A_1,\dots,A_{i-1}, A_{i+1},\ldots, A_l,y)$.
If such an $i$ does not exist, then the new nonterminal $C_j$ is not needed. In order to
keep the notation uniform, let $C_j = B_{i_j-1}$. Finally, we redefine the production for
$Z$ to
$$
Z \to X_1 C_1 X_2 C_2 \cdots  X_k C_k(A).
$$
This concludes the construction of the SLT grammar $G'$.
As in the proof of Theorem~\ref{thm:canon} one can argue that the size of $G'$ is polynomially
bounded in the size of $G$.
\qed
\end{proof}
From  Corollary~\ref{theo:iso-uo}, Lemma~\ref{lemma:bcanon},  and Theorem~\ref{thm:bcanon} we get:

\begin{theorem}
For given SLT grammars $G_1$ and $G_2$ one can check in polynomial time, whether $\val_{\uo}(G_1)$
and $\val_{\uo}(G_2)$ are bisimulation equivalent.
\end{theorem}

\section{Unordered Isomorphism of Non-Linear ST Grammars}
\label{sect:nonlinear}

In this section, we consider ST grammars that are not necessarily linear.

\begin{theorem} \label{thm:non-linear-iso}
The question, whether $\val_{\uo}(G_1)$ and $\val_{\uo}(G_2)$ are isomorphic
for two given  ST grammars $G_1$ and $G_2$ is \PSPACE-hard and in \EXPTIME.
\end{theorem}

\begin{proof}
The upper bound follows from Lemma~\ref{lemma:ST-SLT}  and Corollary~\ref{theo:iso-uo}.
 For the lower bound, we use a reduction from QBF. Recall that the input for QBF is a quantified
 boolean formula of the form 
 \begin{equation} \label{formula:Psi}
 \Psi = Q_1 z_1 Q_2 z_2 \cdots Q_n z_n : \varphi(z_1, \ldots, z_n),
\end{equation}
 where $Q_i \in \{\forall, \exists\}$, the $z_i$ are boolean variables, and $\varphi(z_1, \ldots, z_n)$ 
 is a quantifier-free boolean formula. We can assume that in $\varphi$, negations only occur in front
 of variables.
 We use a reduction from the evaluation problem for boolean expressions to the isomorphism
 problem for explicitly given rooted unordered trees from \cite{JKMT03}.
 Let us take trees $s_1, s_2, t_1, t_2$.
 Consider the two trees $s$ and $t$ in Figure~\ref{fig:and} that are built up from $s_1, s_2, t_1, t_2$.
 Clearly, $s \cong t$ ($s$ and $t$ are isomorphic) if and only if $s_1 \cong t_1$ and $s_2 \cong t_2$.
 Similarly, for the trees $s$ and $t$ from Figure~\ref{fig:or} we have
 $s \cong t$ if and only if $s_1 \cong t_1$ or $s_2 \cong t_2$.
 
 \begin{figure}[tbhp]
    \centering
    \begin{tikzpicture}[node distance=.5cm, every node/.style={minimum height=.6cm}]
        \node (f1) {$f$};
        \node[below left=of f1] (a1) {$a$};
        \node[below right=of f1] (b1) {$b$};
        \begin{scope}[every node/.style={shape=isosceles triangle, shape border rotate=90,
            minimum height=1.5cm}]
            \node[shape=isosceles triangle,draw,below=of a1] (s1) {$s_1$};
            \node[shape=isosceles triangle,draw,below=of b1] (s2) {$s_2$};
        \end{scope}
        \node[below=3 of f1] (c1) {tree $s$};

        \node[right=5 of f1] (f2) {$f$};
        \node[below left=of f2] (a2) {$a$};
        \node[below right=of f2] (b2) {$b$};
        \begin{scope}[every node/.style={shape=isosceles triangle, shape border rotate=90,
            minimum height=1.5cm}]
            \node[shape=isosceles triangle,draw,below=of a2] (t1) {$t_1$};
            \node[shape=isosceles triangle,draw,below=of b2] (t2) {$t_2$};
        \end{scope}
        \node[below=3 of f2] (c2) {tree $t$};

        \draw (f1) -- (a1);
        \draw (f1) -- (b1);
        \draw (a1) -- (s1);
        \draw (b1) -- (s2);
        \draw (f2) -- (a2);
        \draw (f2) -- (b2);
        \draw (a2) -- (t1);
        \draw (b2) -- (t2);
    \end{tikzpicture}
    \caption{The and-gadget}
    \label{fig:and}
 \end{figure}
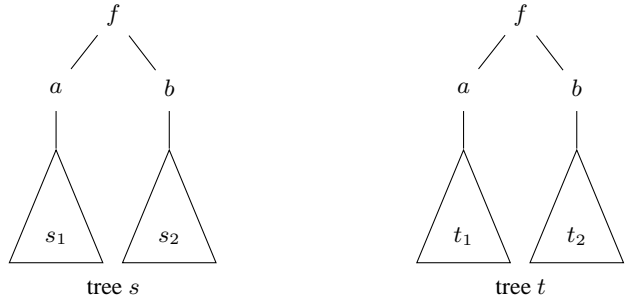
 \begin{figure}[tbhp]
    \centering
    \begin{tikzpicture}[node distance=.5cm, every node/.style={minimum height=.6cm}]
        \node (f11) {$f$};
        \node[below left=.5 and .8 of f11] (f21) {$f$};
        \node[below left=.25 and .2 of f21] (a11) {$a$};
        \node[below right=.25 and .2 of f21] (b11) {$b$};
        \begin{scope}[every node/.style={shape=isosceles triangle, shape border rotate=90,
            minimum height=1.2cm}]
            \node[shape=isosceles triangle,draw,below=of a11] (s11) {$s_1$};
            \node[shape=isosceles triangle,draw,below=of b11] (s12) {$t_2$};
        \end{scope}
        \node[below right=.5 and .8 of f11] (f31) {$f$};
        \node[below left=.25 and .2 of f31] (a12) {$a$};
        \node[below right=.25 and .2 of f31] (b12) {$b$};
        \begin{scope}[every node/.style={shape=isosceles triangle, shape border rotate=90,
            minimum height=1.2cm}]
            \node[shape=isosceles triangle,draw,below=of a12] (t11) {$t_1$};
            \node[shape=isosceles triangle,draw,below=of b12] (t12) {$s_2$};
        \end{scope}
        \node[below=3.7 of f11] (c1) {tree $s$};

        \draw (f11) -- (f21);
        \draw (f11) -- (f31);
        \draw (f21) -- (a11);
        \draw (f21) -- (b11);
        \draw (a11) -- (s11);
        \draw (b11) -- (s12);
        \draw (f31) -- (a12);
        \draw (f31) -- (b12);
        \draw (a12) -- (t11);
        \draw (b12) -- (t12);

        \node[right=5 of f11] (f12) {$f$};
        \node[below left=.5 and 0.8 of f12] (f22) {$f$};
        \node[below left=.25 and .2 of f22] (a21) {$a$};
        \node[below right=.25 and .2 of f22] (b21) {$b$};
        \begin{scope}[every node/.style={shape=isosceles triangle, shape border rotate=90,
            minimum height=1.2cm}]
            \node[shape=isosceles triangle,draw,below=of a21] (s21) {$s_1$};
            \node[shape=isosceles triangle,draw,below=of b21] (s22) {$s_2$};
        \end{scope}
        \node[below right=.5 and .8 of f12] (f32) {$f$};
        \node[below left=.25 and .2 of f32] (a22) {$a$};
        \node[below right=.25 and .2 of f32] (b22) {$b$};
        \begin{scope}[every node/.style={shape=isosceles triangle, shape border rotate=90,
            minimum height=1.2cm}]
            \node[shape=isosceles triangle,draw,below=of a22] (t21) {$t_1$};
            \node[shape=isosceles triangle,draw,below=of b22] (t22) {$t_2$};
        \end{scope}
        \node[below=3.7 of f12] (c1) {tree $t$};

        \draw (f12) -- (f22);
        \draw (f12) -- (f32);
        \draw (f22) -- (a21);
        \draw (f22) -- (b21);
        \draw (a21) -- (s21);
        \draw (b21) -- (s22);
        \draw (f32) -- (a22);
        \draw (f32) -- (b22);
        \draw (a22) -- (t21);
        \draw (b22) -- (t22);
    \end{tikzpicture}
    \caption{The or-gadget}
    \label{fig:or}
 \end{figure}
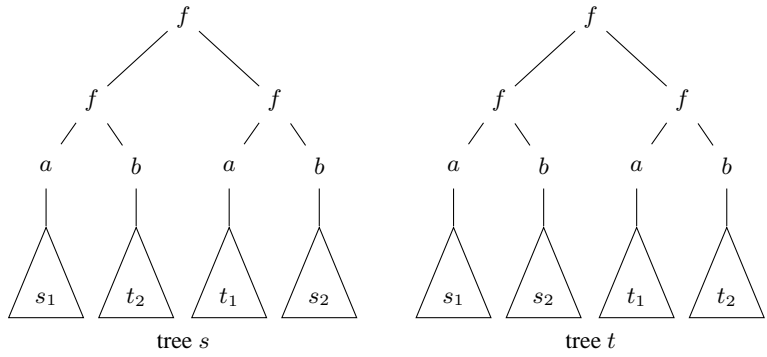

Fix the ranked alphabet $\Sigma=\{f,a,b,0,1\}$.
We will construct a non-linear ST grammar $G$ (without start variable), which contains 
for every subformula $\psi(v_1,\ldots,v_m)$ (where $\{ v_1, \ldots, v_m\} \subseteq 
\{ z_1, \ldots, z_n \}$ is the set of free variables of $\psi$) of the formula $\Psi$ from \eqref{formula:Psi}, two nonterminals
$A_\psi(v_1,\ldots,v_m)$ and $B_\psi(v_1,\ldots,v_m)$ such that for all truth values $c_1, \ldots, c_m \in \{0,1\}$:
$\psi(c_1,\ldots,c_m)$ evaluates to $1$ if and only if $\val_G(A_\psi)[ v_i \gets c_i \mid i \in [m]]$ and 
$\val_G(B_\psi)[ v_i \gets c_i \mid i \in [m]]$ are isomorphic as rooted unordered trees.

The base case is that of a literal $z$ or $\neg z$.
We introduce the following productions:
$$
A_{z}(z) \to f(z,1), \quad B_{z}(z) \to f(1,z), \quad A_{\neg z}(z) \to f(z,0), \quad B_{\neg z}(z) \to f(0,z) .
$$
Now let 
$\psi(v_1, \ldots, v_m) = \psi_1(x_1, \ldots, x_k) \wedge \psi_2(y_1, \ldots, y_l)$ be a subformula
of the quan\-tifier-free part $\varphi(z_1, \ldots, z_n)$ in \eqref{formula:Psi}, 
where $\{ v_1, \ldots, v_m\} = \{x_1, \ldots, x_k, y_1, \ldots, y_l\}$. Then 
we use the and-gadget from Figure~\ref{fig:and} and set (where $\bar v = (v_1, \ldots, v_m)$
and similarly for $\bar x$ and $\bar y$)
\begin{eqnarray*}
A_{\psi}(\bar v) & \to & f( a(A_{\psi_1}(\bar x)), b(A_{\psi_2}(\bar y))) \text{ and }\\
B_{\psi}(\bar v) & \to & f( a(B_{\psi_1}(\bar x)), b(B_{\psi_2}(\bar y))).
\end{eqnarray*}
If  $\psi(v_1, \ldots, v_m) = \psi_1(x_1, \ldots, x_k) \vee \psi_2(y_1, \ldots, y_l)$
then we use the or-gadget from Figure~\ref{fig:or} and set
\begin{eqnarray*}
A_{\psi}(\bar v) & \to & f( f(a(A_{\psi_1}(\bar x)), b(B_{\psi_2}(\bar y))),  f(a(B_{\psi_1}(\bar x)), b(A_{\psi_2}(\bar y))))  \text{ and }\\
B_{\psi}(\bar v) & \to & f( f(a(A_{\psi_1}(\bar x)), b(A_{\psi_2}(\bar y))),  f(a(B_{\psi_1}(\bar x)), b(B_{\psi_2}(\bar y)))) .
\end{eqnarray*}
For a quantified subformula  $\psi(z_1, \ldots, z_{i-1}) = \forall z_i \, \psi'(z_1, \ldots, z_{i-1}, z_i)$, we can define 
the productions similarly (let $\bar z = (z_1, \ldots, z_{i-1})$):
\begin{eqnarray*}
A_{\psi}(\bar z) & \to & f( a(A_{\psi'}(\bar z,0)), b(A_{\psi'}(\bar z,1))) \text{ and }\\
B_{\psi}(\bar z) & \to & f( a(B_{\psi'}(\bar z,0)), b(B_{\psi'}(\bar z,1))).
\end{eqnarray*}
Finally, for $\psi(z_1, \ldots, z_{i-1}) = \exists z_i \, \psi'(z_1, \ldots, z_{i-1}, z_i)$ we set
\begin{eqnarray*}
A_{\psi}(\bar z) & \to & f( f(a(A_{\psi'}(\bar z,0)), b(B_{\psi'}(\bar z,1))),  f(a(B_{\psi'}(\bar z,0)), b(A_{\psi'}(\bar z,1))))  \text{ and }\\
B_{\psi}(\bar z) & \to & f( f(a(A_{\psi'}(\bar z,0)), b(A_{\psi'}(\bar z,1))),  f(a(B_{\psi'}(\bar z,0)), b(B_{\psi'}(\bar z,1)))) .
\end{eqnarray*}
This concludes the construction of the ST grammar $G$. Let $G = (N,\Sigma,P)$.
Then we define the two ST grammars $G_1 = (N,\Sigma, A_{\Psi}, P)$ and 
$G_2 = (N,\Sigma, B_{\Psi}, P)$. We have $\val_{\uo}(G_1) \cong \val_{\uo}(G_2)$ if and only
if the formula $\Psi$ is true.
 \qed
\end{proof}
The complexity bounds from
Theorem~\ref{thm:non-linear-iso} also hold if we want to check whether the unrooted unordered trees
 $\val_{\ur,\uo}(G_1)$ and $\val_{\ur,\uo}(G_2)$ are isomorphic: Membership in \EXPTIME{} follows from
Lemma~\ref{lemma:ST-SLT}  and Corollary~\ref{coro-ur-uo-iso}.
 For \PSPACE-hardness, one can take the reduction from the 
proof of Theorem~\ref{thm:non-linear-iso} and label the roots of the final trees
with a fresh symbol.
Finally, the above \PSPACE-hardness proof can be also used for the bisimulation equivalence
problem for trees given by ST grammars (the gadgets from  Figure~\ref{fig:and} and~\ref{fig:or}
can be reused). Hence, bisimulation equivalence for trees given by ST grammars is 
\PSPACE-hard and in \EXPTIME. Since an ST grammar can be transformed into a 
hierarchical graph definition for a dag (see the proof of Lemma~\ref{lemma:ST-SLT}),
we rediscover the following result from  \cite{DBLP:conf/concur/BrenguierGS12}:
Bisimulation equivalence for dags that are given by hierarchical graph definitions is
\PSPACE-hard and in \EXPTIME.

\section{Open problems}

The obvious remaining open problem is the precise complexity of the isomorphism problem for unordered trees that
are given by ST grammars. Theorem~\ref{thm:non-linear-iso} leaves a gap from \PSPACE{} to \EXPTIME{}.
Another interesting open problem is the isomorphism problem for graphs that are given by hierarchical graph definitions.
To the knowledge of the authors, this problem has not been studied so far.


\end{document}